\documentclass[conference]{IEEEtran}
\usepackage{subfigure}
\usepackage{arydshln}
\usepackage[cmex10]{amsmath}
\usepackage{amssymb}
\usepackage{bm}
\usepackage{mathrsfs}
\usepackage{caption}
\usepackage[T1]{fontenc}

\DeclareMathOperator{\Span}{span}
\DeclareMathOperator{\Diag}{diag}
\DeclareMathOperator{\Vol}{vol}
\DeclareMathOperator{\Rank}{rank}
\DeclareMathOperator{\Trace}{Tr}

\DeclareMathOperator{\Dim}{dim}
\DeclareMathOperator{\Corr}{corr}

\newtheorem{Def}{\textbf{Definition}}
\newtheorem{Problem}{\textbf{Problem}}
\newtheorem{Theorem}{\textbf{Theorem}}
\newtheorem{Lemma}{\textbf{Lemma}}

\usepackage[nospace,nocompress]{cite}

\ifCLASSINFOpdf
   \usepackage[pdftex]{graphicx}
   \DeclareGraphicsExtensions{.pdf,.jpeg,.png}
\else
   \usepackage{graphicx}
   \DeclareGraphicsExtensions{.eps}
\fi
\begin{document}
%
\title{A Volume Correlation Subspace Detector for signals buried in unknown clutter}
	\author{Hailong Shi, Hao Zhang, and Xiqin Wang \\Intelligent Sensing Lab, Department of Electronic Engineering, Tsinghua University%
		\thanks{Intelligent Sensing Lab, Department of Electronic Engineering, Tsinghua University, Beijing, E-mail: shl06@mails.tsinghua.edu.cn, haozhang@mail.tsinghua.edu.cn, wangxq\_ee@tsinghua.edu.cn. This paper is supported by project 61201356 of National Natural Science Foundation of China.}}



%


\maketitle

\begin{abstract}
Detecting the presence of target subspace signals with unknown clutters is a well-known hard problem encountered in various signal processing applications. Traditional methods fails to solve this problem because prior knowledge of clutter subspace is required, which can not be obtained when target and clutter are intimately mixed. In this paper, we propose a novel subspace detector that can detect target signal buried in clutter without knowledge of clutter subspace. This detector makes use of the geometrical relation between target and clutter subspaces and is derived based upon the calculation of volume of high dimensional geometrical objects. Moreover, the proposed detector can accomplish the detection simultaneously with the learning processes of clutter, a property called "detecting while learning". The performance of detector was showed by theoretical analysis and numerical simulation.
\end{abstract}



%
\IEEEpeerreviewmaketitle

\section{Introduction}
\par
We consider the following subspace signal detection problem widely existing in communication, radar, sonar and other fields of signal processing:

\begin{Problem}\label{Pro3}
	Let $\bm{H}$ be a Hilbert space, $\bm{H}_S\subset\bm{H}$ be a KNOWN signal subspace and $\bm{H}_C\subset\bm{H}$ be an UNKNOWN clutter subspace. Given the sampled data $\bm{y}\in\bm{H}$, could we determine whether $\bm y$ lies in $\bm{H}_S\oplus\bm{H}_C$ or entirely in $\bm{H}_C$ with the influence of random noise? In other words, whether $\bm y$ satisfies
\begin{equation}
	\bm y=\bm s+ \bm c + \bm w, 
\end{equation}
or
\begin{equation}
  \bm y= \bm c + \bm w,
\end{equation}
where $\bm s\in\bm{H}_S,\ \bm c\in\bm{H}_C,\ \bm s, \bm c\neq 0$ and $w$ is random noise.
\end{Problem}

Detecting target signal in certain signal subspace with known clutter has been considered by several researchers and various schemes has been proposed. Among these works, the clutter subspace is commonly modeled to have low rank\cite{scharf1994matched,raghavan2012statistical}, and the most remarkable approach is the Matched Subspace Detector (MSD) by Scharf et. al. \cite{scharf1994matched}, which is actually a generalized energy detector using the prior knowledge of target and clutter subspaces. Although tremendous variations and applications of MSD has appeared \cite{scharf1996adaptive,kraut1999cfar,kraut2001adaptive,desai2003robust}, the key precondition for the success of matched subspace detector is that the clutter subspace $\bm{H}_C$ must be KNOWN beforehand. It usually can not be satisfied in practice, e.g., in radar, reconnaissance,  mobile communication, etc. As far as we know, because the projection-based detectors derived from GLRT can not be constructed explicitly, there hasn't been any detector that can work when the clutter subspace is totally unknown.
\par
However, the structure of clutter subspace could be explored by successively sampling from it. As a matter of fact, a generic property of randomly sampling in linear space is that the sampled vectors are generally linearly independent, so ideally the clutter subspace could be "reconstructed" by obtaining multiple samples from the clutter subspace, which naturally form the basis of the clutter subspace. But it should be noted that generally the information of target signal is mixed intimately with the clutter in the samples and it is impossible to separated the "pure" clutter from target signal so that the basis for clutter subspace could be extracted alone. It is the core difficulty for detecting the target signal against unknown structured clutters. 

In this paper, a novel method to detect target signal embedded in an unknown structured low-rank clutter was given. The main idea of our detector is utilizing the geometrical characteristic of the sampled data. Here the volume, which is a common concept for geometrical objects, is  defined to measure relationships between subspaces (more concretely, it is the parallelotope with its edges being the basis vectors of subspace). An intuition of using volume for detection is that, since the "volume" of low-dimensional subspaces in high-dimensional linear space is zero, the judgment of whether or not a target vector lies in certain subspace could be transformed naturally to the calculation of "volume" of parallelotope built by the target vector and the basis vectors of that subspace, in space with pertinent dimension. If the 'volume' is zero, then the conclusion can be drawn that the subspace contains the target vector, otherwise the target vector must lie outside the subspace. 
Indeed, Although there have been researches using similar geometrical approaches to perform subspace signal detection\cite{ramirez2010detection,cochran1995geometric}, the advantage that the proposed volume-based detector requires no knowledge of clutter is first discovered and analyzed by us.

The remainder of this paper is organized as follows: Some preliminary backgrounds on the geometrical concepts for linear subspaces such as principal angles and volumes were summarized in section II. Then the volume-based subspace detector was introduced and its property was discussed in detail in section III and IV respectively.

\section{Preliminary Background}

In this section, some important concepts of linear space geometry were reviewed concisely. Only necessary material for our discussion was put forward for the space limitation. For details, please see \cite{absil2004riemannian} and reference therein.

\subsection{Principal Angles between Subspaces}
\par
The concept of principal angles \cite{miao1992principal} is the natural generalization of that of angles between two vectors. Principal angles can be used to formulate the relationship between two subspaces.

\begin{Def}
	For two linear subspaces $\bm{H}_1$ and $\bm{H}_2$, with dimensions $\dim(\bm{H}_1)=d_1,\dim(\bm{H}_2)=d_2$. Take $m = \min(d_1,d_2)$, then the principal angles $0 \leq \theta_1 \leq \cdots \leq \theta_m \leq \pi/2$ between $\bm{H}_1$ and $\bm{H}_2$ are defined by
	\begin{eqnarray}
	\cos \theta_i = \max_{\bm u_i \in \bm{H}_1,\bm v_i \in \bm{H}_2}\bm u_i^T \bm v_i,\qquad \textit{subject to} \nonumber \\
	\|\bm u_i\|_2=\|\bm v_i\|_2=1,\qquad
		\bm u_i^T \bm u_j = 0, \bm v_i^T \bm v_j=0,\quad \nonumber
	\end{eqnarray}
where $j = 1,\cdots,i-1, i = 1,\cdots m$.
\end{Def}

\par
As important concepts of linear space geometry, the principal angles are widely applied in scientific and engineering fields. The geodesic distance which is the key metric measure on Grassmann manifold, as well as numerous kinds of distance measures, is defined using the principal angles \cite{absil2004riemannian,qiu2005unitarily}, such as the Chordal Distance, Binet-Cauchy Distance and Procrustes Distance commonly seen in signal process applications \cite{hamm2008grassmann,hamm2008subspace}.
Moreover, the volume of subspace used in this paper to construct our subspace detector is also closely related to the principal angles.

\subsection{The Volume of a matrix}

\par
 Suppose a full-rank matrix $\bm{X}\in\mathbb{R}^{n\times{d}}$, $d<n$, then its $d$-dimensional volume is defined as \cite{ben1992volume}

\begin{equation}\label{VolumeDef1}
\Vol_d(\bm X):= \prod_{i=1}^d \sigma_i,
\end{equation}
where $\sigma_1 \geq \sigma_2 \geq \cdots \geq \sigma_d > 0$ are the non-zero singular values of $\bm{X}$. For $\bm X$ is of full column rank, its  $d$-dimensional volume can be written equivalently as \cite{miao1992principal,ben1992volume}
\begin{equation}\label{VolumeDef2}
\Vol_d (\bm X) = \sqrt{\det(\bm{X^TX})}.
\end{equation}

The following simple lemma is widely useful in application of volume for subspaces. It means that the $k$-dimensional volume of the basis of subspace with dimension less than $k$ is definitely zero.

\begin{Lemma}\label{lemma2}
	Suppose $\bm{X}^{(n)}=[\bm{x}_1,\cdots,\bm{x}_n]$ be a matrix whose columns are a group of vectors in a Hilbert space $\bm H$ and $\dim(\Span{(\bm{X}^{(n)})})=i$, then
	\begin{equation}
	\Vol_{k}(\bm{X}^{(n)})=0,\quad{i<k}
	\end{equation}
\end{Lemma}

$d$-dimensional volume provides a kind of measure of separation between two linear subspaces. Normalized by individual volumes of both matrices, the volume of a matrix composed of two matrices gives a new kind of correlation, called the volume correlation. For $n$-dimensional Hilbert space $\bm{H}$ and its two subspaces $\bm{H}_1$ and $\bm{H}_2$ with dimensions $\dim(\bm{H}_1)=d_1,\dim(\bm{H}_2)=d_2$, we have
\begin{equation}\label{CorrVol}
\Corr_{\textrm{vol}}(\bm{H}_1,\bm{H}_2)=\frac{\Vol_{d_1+d_2}([\bm{X}_1,\bm{X}_2])}{\Vol_{d_1}(\bm{X}_1)\Vol_{d_2}(\bm{X}_2)},
\end{equation}
where $\bm X_1$ and $\bm X_2$ are basis matrices of $\bm H_1$ and $\bm H_2$, and $[\bm{X}_1,\bm{X}_2]$ means putting columns of matrices $\bm X_1$ and $\bm X_2$ together. It is closely related to the principal angles between subspaces \cite{miao1992principal},
\begin{equation}\label{VolAng}
\Corr_{\textrm{vol}}(\bm{H}_1,\bm{H}_2) = \prod_{j=1}^{\min(d_1,d_2)}\sin\theta_j(\bm{H}_1, \bm{H}_2),
\end{equation}
where $0\leq\theta_j(\bm{H}_1, \bm{H}_2)\leq2\pi,1\leq{j}\leq\min(d_1,d_2)$ are the principal angles of subspaces $\bm{H}_1$ and $\bm{H}_2$.

\par
It can be seen intuitively from (\ref{VolAng}) that the volume correlation $\Corr_{\textrm{vol}}(\bm{H}_1,\bm{H}_2)$ can actually play the role of distance measure between subspaces $\bm{H}_1$ and $\bm{H}_2$. When $\bm{H}_1$ and $\bm{H}_2$ have vectors in common, i.e., $\Dim(\bm{H}_1\bigcap\bm{H}_2)\geq1$, we have $\Corr_{\textrm{vol}}(\bm{H}_1,\bm{H}_2)=0$. On the other side, when $\bm{H}_1$ is orthogonal to $\bm{H}_2$, we have $\Vol_{d_1+d_2}([\bm X_1, \bm X_2])=\Vol_{d_1}(\bm X_1)\Vol_{d_2}(\bm X_2)$, in other words, $\Corr_{\textrm{vol}}(\bm{H}_1,\bm{H}_2)=1$. Although volume correlation may not rigorously be a metric measure, we still regard it as a generalized distance measure that plays a key role in our proposed subspace detector.

\section{The Correlation Subspace Detector in a noiseless environment}

From this section, we are going to introduce the volume correlation subspace detector (or VC subspace detector briefly) step by step. In order to fully convey the geometrical intuition about our subspace detector, in this section, we temporarily assume the noise component is not present, i,e., $\bm w=0$ in Problem 1. For clearance and easy of understanding, we just give some main idea about the geometrical explanation of the proposed detector in noiseless environment, and leave rigorous analysis of the noisy situation for the next section.

\subsection{Main Idea}

Unknown clutters with subspace structures were the primary obstacle for efficient detection of target signal. To reach the purpose, the designer of detectors must find the way to clarify the intrinsic construction of clutter subspace. Just as most of the traditional approaches for background learning, multiple samples were adopted to explore the clutter subspace. The following observation is the foundation for the exploration of clutter subspace.

\begin{itemize}
	\item Suppose $\bm H$ be a $n$-dimensional Hilbert space, $\bm{x}_1,\cdots,\bm{x}_k$, $k<n$ be randomly sampled vectors from $\bm H$, then in the generic situation, we have
	\begin{equation}
	\Dim(\Span\{\bm{x}_1,\cdots,\bm{x}_k\})=k,
	\end{equation}
	In words, random samples $\bm{x}_1,\cdots,\bm{x}_k$ are generally linearly independent.
	\item In the case of $k\geq{n}$, then in the generic situation, we have
	\begin{equation}
	\Dim(\Span\{\bm{x}_1,\cdots,\bm{x}_k\})=n,
	\end{equation}
	In words, $\bm{x}_1,\cdots,\bm{x}_k$ are linearly dependent.
\end{itemize}

Let $\bm{H}_C$ be the unknown clutter subspace with unknown dimension $d_1$, $\bm{H}_S$ be the known target subspace with dimension $d_2$, $\bm{y}_1,\cdots,\bm{y}_{d_1}$ be samples representing our sample subspace. Without loss of generality, we assume
$\Dim(\bm{H}_S\bigcap\bm{H}_C)=0$ throughout this paper. The critical point when we explore the clutter $\bm{H}_C$ is that, the sample subspace may contain both clutter and target signals in general. In other words, we cannot guarantee that the sample subspace is a "pure" clutter subspace. What we commonly get are samples like:
\begin{equation}
\bm{y}_i = \bm{s}_i + \bm{c}_i, \quad{\bm{s}_i\in\bm{H}_S,\ \bm{c}_i\in\bm{H}_C},\ i=1,2,\cdots,
\end{equation}
It is impossible to separate the clutter and target signal apart and build the clutter subspace from these $\bm y_i$. How does the information of sample subspace be mined effectively?

The volume correlation between subspaces is helpful for us to eliminate the impact of mixing of clutter and target signal. In fact, Let $\bar{\bm s}_1,\cdots,\bar{\bm s}_{d_2}$ be the known basis vectors of $\bm{H}_S$. 
It has been mentioned that in the generic scenario, different $\bm{y}_i$ sampled from $\bm H_S \oplus\bm{H}_C$ are linearly independent. In other words, innovative directions of basis vectors in $\bm H_S \oplus\bm{H}_C$ are revealed continually along with the sampling process as follows,
\begin{align}
&\Dim(\Span\{\bm{y}_1,\bar{\bm s}_1,\cdots,\bar{\bm s}_{d_2}\})=1+d_2,\nonumber\\
&\Dim(\Span\{\bm{y}_1,\bm{y}_2,\bar{\bm s}_1,\cdots,\bar{\bm s}_{d_2}\})=2+d_2,\nonumber\\
&\cdots\cdots\nonumber\\
&\Dim(\Span\{\bm{y}_1,\cdots,\bm{y}_{d_1},\bar{\bm s}_1,\cdots,\bar{\bm s}_{d_2}\})=d_1+d_2,
\end{align}
The question is what will happen next. 
The above process could be analyzed in another way from the viewpoint of volume mentioned earlier. In particular, the core idea of our subspace detector in noiseless environment could be illustrated fully using volume. Firstly, when both the signal and clutter are present, we have
\begin{align}
&\Vol_{1+d_2}([\bm{y}_1,\bar{\bm s}_1,\cdots,\bar{\bm s}_{d_2}])>0,\nonumber\\
&\Vol_{2+d_2}([\bm{y}_1,\bm{y}_2,\bar{\bm s}_1,\cdots,\bar{\bm s}_{d_2}])>0,\nonumber\\
&\cdots\cdots\nonumber\\
&\Vol_{d_1+d_2}([\bm{y}_1,\cdots,\bm{y}_{d_1},\bar{\bm s}_1,\cdots,\bar{\bm s}_{d_2}])>0,
\end{align}
Next, the magic will happen for the next dimension, i.e., when there are $d_1+d_2+1$ sample vectors, there will be
\begin{equation}
\Vol_{d_1+d_2+1}([\bm{y}_1,\cdots,\bm{y}_{d_1},\bm{y}_{d_1+1},\bar{\bm s}_1,\cdots,\bar{\bm s}_{d_2}])=0,\label{volscc}
\end{equation}
while on the other hand,
\begin{equation}
\Vol_{d_1+1}(\Span\{\bm{y}_1,\cdots,\bm{y}_{d_1},\bm{y}_{d_1+1}\})>0.\label{volsc}
\end{equation}
The reason is because, $d_1+d_2+1$ sample vectors in this scenario have not spanned the entire subspace $\bm H_S \oplus \bm H_C$ according to the previous statement of randomly sampling; but $\bm{y}_1,\cdots,\bm{y}_{d_1},\bm{y}_{d_1+1},\bar{\bm s}_1,\cdots,\bar{\bm s}_{d_2}$ can span $\bm H_S \oplus \bm H_C$, in another word,
\begin{eqnarray}
\Dim(\Span\{\bm{y}_1,\cdots,\bm{y}_{d_1},\bm{y}_{d_1+1},\bar{\bm s}_1,\cdots,\bar{\bm s}_{d_2}\})\nonumber \\
=\Dim(\bm{X}_S\oplus\bm{X}_C)=d_1+d_2, \hspace{2cm}
\end{eqnarray}
On the other hand,  if the sample subspace only contains pure clutter, we obtain
\begin{equation}
\Vol_{d_1+d_2+1}([\bm{y}_1,\cdots,\bm{y}_{d_1},\bm{y}_{d_1+1},\bar{\bm s}_1,\cdots,\bar{\bm s}_{d_2}])=0,\label{volcc}
\end{equation}
and
\begin{equation}
\Vol_{d_1+1}(\Span\{\bm{y}_1,\cdots,\bm{y}_{d_1},\bm{y}_{d_1+1}\})=0.\label{volc}
\end{equation}
(\ref{volscc}), (\ref{volsc}), (\ref{volcc}) and (\ref{volc}) indicates that, $d_1+1$ is the critical number of samples for detection of target signal in the background of clutter with unknown subspace structure, i.e., the "breakpoint". The knack of detection in this noiseless situation is, sampling continually, computing the volume of parallelotope spanned by all the sample vectors and known basis of target subspace at various dimensions and inspect the change of results. Once the volume vanishes, it means the number of samples reaches the critical point. Then the process of sampling should be stopped and the volume of sample vector themselves is calculated. The decision can be made based on whether the result is zero, i.e., whether (\ref{volsc}) or (\ref{volc}).

\section{The Volume-Correlation Subspace Detector in noisy environment}
\subsection{Main Idea}
\par
The main problem here is the sample subspace has been contaminated by random noise and can not be used directly to compute the volume correlation in VC subspace detector. Therefore the noise must be cleared in advance. For most statistical signal processing algorithms concerned with subspaces, such as MUSIC, ESPRIT and so on, the target signal and random noise are separated into signal subspaces and noise subspaces by eigen-decomposition of correlation matrices firstly for further treatment. It implies the natural strategy of extracting signal subspaces for follow-up analysis and discarding noise subspaces simply for noise elimination.

To be specific, we reconsider Problem 1 where $\bm w$ is assumed to be white Gaussian noise with zero mean and variance $\sigma^2$. Traditional subspace methods mentioned above deal with the correlation matrix $\mathrm{R}_{\bm y}$ of the sample data $\bm{y}$, which is denoted by
\begin{equation}\label{CorrMatrix}
\mathrm{R}_{\bm y}=\mathbb{E}\{\bm{y}\bm{y}^T\}
\end{equation}
According to our problem setting, the eigenvalues of $\mathrm{R}_{\bm y}$ could be listed as
\begin{equation}\label{realEigenvalue}
\lambda_1\geq\lambda_2\geq\cdots\geq\lambda_k\geq\lambda_{k+1}=\cdots=\lambda_n=\sigma^2,
\end{equation}
and the corresponding eigenvectors are
$$
\bm{q}_1, \bm{q}_2,\cdots,\bm{q}_k,\bm{q}_{k+1},\cdots,\bm{q}_n
$$
Denote $\bm{Q}_{SC}:=[\bm{q}_1, \bm{q}_2,\cdots, \bm{q}_k]\in\mathbb{R}^{n\times{k}}$, $\bm{Q}_{N}:=[\bm{q}_{k+1}, \bm{q}_{k+2},\cdots, \bm{q}_n]\in\mathbb{R}^{n\times(n-k)}$. It is clear that
\begin{equation}
\Span(\bm{Q}_{SC}) = \bm H_S \oplus \bm H_C,\quad k=d_1+d_2,
\end{equation}
when both signal and clutter are present, and
\begin{equation}
\Span(\bm{Q}_{SC}) =  \bm H_C,\quad k=d_1,
\end{equation}
when the sampled data contains "pure" clutter. The
$\Span(\bm{Q}_{SC})$ and $\Span(\bm{Q}_{N})$ are commonly called signal subspace and noise subspace. Asymptotically $\Span(\bm{Q}_{SC})$ could be used as proxy of $\bm{H}_S\oplus\bm{H}_C$ (or $\bm H_C$) and the main idea in previous section is workable as well in the noisy environment.

\subsection{The proposed VC subspace Detector}

The VC subspace detector is extended to noisy scenario as follows:

\begin{itemize}
	
	\item \textbf{Initial Step :}
	Denote the received data $\{\bm{y}_1,\cdots,\bm{y}_n\}$ by $\bm{R}^{(n)}$. Obtain $\{{\bm s}_1,\cdots,{\bm s}_{d_2}\}$ as the orthonormal basis vectors of known target subspace and denote it by $\bm{Q}_S$. Let the sample covariance matrix be $\hat{\mathrm{R}}^{(0)}=0$. Index $i$ is set to $1$. Set two thresholds $T$ and $\epsilon$ at appropriate values.
	
	\item \textbf{Step 1} :
	Get the new sample $\bm{y}_{i}$, compute the covariance matrix as
	\begin{equation}\label{CorrMat}
	\hat{\mathrm{R}}^{(i)}=\frac{i-1}{i}\hat{\mathrm{R}}^{(i-1)}+\frac{1}{i}\bm{y}_{i}\bm{y}_{i}^{\rm T},
	\end{equation}
	Assume the eigenvalues of $\hat{\mathrm{R}}^{(i)}$ be
	\begin{equation}\label{eigenvalue}
	\hat{\lambda}_1\geq\hat{\lambda}_2\geq\cdots\geq\hat{\lambda}_{k_i}\geq\hat{\lambda}_{k_i+1}=\cdots=\hat{\lambda}_n,
	\end{equation}
	and the corresponding eigenvectors be
	\begin{equation}\label{eigenvector}
	\hat{\bm{q}}_1, \hat{\bm{q}}_2,\cdots,\hat{\bm{q}}_{k_i},\hat{\bm{q}}_{k_i+1},\cdots,\hat{\bm{q}}_n,
	\end{equation}
	obtain the estimated basis of the sampled subspace by
	\begin{equation}
	\hat{\bm{Q}}^{(i)}=[\hat{\bm{q}}_1, \hat{\bm{q}}_2,\cdots,\hat{\bm{q}}_{k_i}],
	\end{equation}
	
	\item \textbf{Step 2 :}
	Compute the test quantity as
	\begin{equation}
	T(\bm{R}^{(i)})=\Vol_{k_i+d_2}([\hat{\bm{Q}}^{(i)},\ \bm{Q}_S]),
	\end{equation}

	\item \textbf{Step 3 :} 
	if $1/T(\bm{R}^{(i)})>T$, conclude the existence of target signal and exit;
	
	else if $|1/T(\bm{R}^{(i)})-1/T(\bm{R}^{(i-1)})|<\epsilon$, conclude the non-existence of target signal and exit;

	otherwise, set $i=i+1$ and go to step 1;
	
\end{itemize}

Remark 1. It should be emphasized that the most remarkable advantage of VC subspace detector is its feature of "Detecting while Learning". To be specific, the detection could be completed without separated sessions for background learning with VC subspace detector. As well known, background learning is very popular in adaptive processing for radar, communication and other signal processing problems. Channel equalization in communication transmission, CFAR (Constant False Alarm Rate) operation in radar detection and estimation of covariance matrices for clutter echoes in STAP (Space-Time Adaptive Processing) all belong to sessions for background learning. There are double common defects for all these schemes. The first is that the efficacy for estimation of clutter background might be influenced heavily by existence of target signal, so called as target leakage in literatures; the second is the non-homogeneousness widely existed in clutter environment which easily leads to mismatch of learning consequence with the actual clutter scenario at the target location. Nevertheless, VC subspace detector stands far away from these trouble because the process of background process is accomplished implicitly and simultaneously with the detection operation. Along with the raw data being sampled and put into work sequentially, the volume correlations are examined and tested constantly until the threshold is reached. The information of clutter subspace is being learned in the form of volumes of low-dimensional approximations of clutter subspace. At the decision point, background learning is ended spontaneously and the decision on the existence of target will be made naturally. There is no need for extra effort of background learning. The learning and detection is merged perfectly in VC subspace detector. We call this interesting property "Detection while Learning". Our VC subspace detector could be listed as blind detecting methods.

Remark 2.
It should be noted that the subspace $\hat{\bm{Q}}^{(i)}$ is actually an estimation of the real signal subspace. The accuracy of this approximation had been studied extensively \cite{stoica1989music,jeffries1985asymptotic} and the feasibility of $\hat{\bm{Q}}^{(i)}$ had been proved asymptotically. Hence we can expect the proposed VC subspace detector will asymptotically approximate the VC subspace detector in noiseless scenario, and this expectation is validated by the theory in next section.

Remark 3.
The dimension of signal-plus-clutter subspace (or clutter subspace), i.e., $k_i$ in (\ref{eigenvalue}) actually needs to be estimated. Since there are various methods can be used, like AIC or MDL\cite{akaike1974new,wax1985detection}, we will not discuss this topic in detail here.

\subsection{Theoretical Property of the VC Subspace Detector}

To avoid the vagueness brought by asymptotical conclusion of the performance of VC subspace detector in the noisy background, we give some non-asymptotical analysis on the capability of our detector with knowledge of random matrices and concentration inequalities.  We have the following result.

\begin{Theorem}\label{H1Thm1N}
	Let $H$ be $n$-dimensional Hilbert space, $\bm{H}_S$ and $\bm{H}_C$ be target and clutter subspaces of $H$ respectively, $\bm{H}_S\cap\bm{H}_C=\{0\}$, $\Dim{\bm{H}_C}=d_1$, $\Dim{\bm{H}_S}=d_2$,
\begin{equation}
\bm{y}_i=\bm{x}_i+\bm{w}_i,\qquad i=1,2,\cdots,m.\nonumber
\end{equation}
where $\bm y_i$ is the sampled data, $\bm x_i\in\bm{H}_S\oplus\bm{H}_C$ (or $\bm H_c$), $\bm{w}_i\thicksim\mathcal{N}(0,\sigma^2\mathrm{I}_n)$ is Gaussian white noise. Denote the eigenvalues of the covariance matrix of received signal $\mathrm{R}_{\bm y}$ by (\ref{realEigenvalue}),

	If the target signal presents in sample data, then for any $0<\varepsilon<1$ and $\delta>0$, if
	\begin{eqnarray}\label{Res8}
	m\geq\frac{1+\varepsilon}{(\sqrt{\delta+1}-1)^2}\cdot \hspace{5cm}\nonumber \\
	\left(\sum_{\stackrel{i,j=1}{j\neq i}}^{d_2+d_1} \frac{\lambda_i\lambda_j}{(\lambda_i-\lambda_j)^2}+(n-d_1-d_2)\sum_{i=1}^{d_2+d_1}\frac{\lambda_i\sigma^2}{(\sigma^2-\lambda_i)^2}\right),
	\end{eqnarray}
	then there exists a constant $C>0$, such that
	\begin{equation}\label{Res7}
	|T(\bm{R}^{(m)})^2|\leq\delta^{d_2} + O(\delta^{d_2+1}).
	\end{equation}
	holds with probability
	\begin{equation}\label{prob1}
	\mathbb{P}\geq1-\exp\{-\frac{(d_1+d_2)\cdot{n}\cdot\varepsilon^2}{C}\}.
	\end{equation}
	
	On the contrary, in the case of non-target, for any $0<\varepsilon<1$ and $\delta>0$, when
	\begin{eqnarray}\label{ResH0m}
	m \geq\frac{1+\varepsilon}{(\sqrt{\delta+1}-1)^2}\cdot \hspace{5cm}\nonumber \\\left( \Big( \sum_{i =1}^{d_1} \sum_{\stackrel{j=1}{j\neq i}}^{d_1} \frac{\lambda_i \lambda_j}{(\lambda_i-\lambda_j)^2}+\sum_{i =1}^{d_1}(n-d_1)\frac{\lambda_i\sigma^2}{(\sigma^2-\lambda_i)^2}  \Big)\right),
	\end{eqnarray}
	we have
	\begin{equation}\label{ResH0}
	|T(\bm R^{(m)})^2-\tau^2(\bm{H}_S,\bm{H}_C)|\leq s_{d_1-1}(\bm{Q}_C^T\bm{P}_S^{\perp}\bm{Q}_C)\delta + O(\delta^{2}),
	\end{equation}
	holds with probability
	\begin{equation}\label{prob2}
	\mathbb{P} \geq 1-\exp\{-\frac{d_1\cdot{n}\cdot \varepsilon^2}{C}\},
	\end{equation}
here $\tau(\bm{H}_S,\bm{H}_C) = \Vol_{d_1+d_2}([\bm Q_S, \bm Q_C])>0$ is a constant related with $\bm{H}_S$ and $\bm{H}_C$, $\bm Q_S$, $\bm{Q}_C$ are the orthogonal bases of $\bm H_S$ and $\bm H_C$, respectively, and $\bm{P}_S^{\perp}$ is the projection matrix onto $\bm{H}_S^{\perp}$, $s_{k}(\bm A)$ for matrix $\bm A \in \mathbb{R}^{n \times n}$ is defined as:
\begin{equation}\label{ESF}
s_{k}(\bm A) := \sum_{1 \leq i_1 \leq \cdots \leq i_k \leq n} \sigma_{i_1}\cdots \sigma_{i_k}, 1\leq k \leq n.
\end{equation}	
\end{Theorem}

\par
Theorem \ref{H1Thm1N} describes the performance of our VC subspace detector in noisy environment. The main result (\ref{Res7}), together with (\ref{Res8}) and (\ref{prob1}), implies that when the target signal is present, the test quantity $1/T(\bm R^{(m)})$ of VC subspace detector will tend to infinity with an overwhelming probability, when the number of sample data is sufficient large. On the other hand, (\ref{ResH0}) together with (\ref{ResH0m}) and (\ref{prob2}) ensures $1/T(\bm R^{(m)})$ to tend to a finite value. Therefore the decision point of this detector is to test whether $1/T(\bm R^{(m)})$ increases over a threshold, or stops increasing at a finite value. The result of Theorem \ref{H1Thm1N} implies that the output of our VC subspace detector will remain almost unchanged no matter what clutter is given, whether or not there is noise, so far as that we have enough sample data. This shows the asymptotic effectiveness of our VC subspace detector.

\par
The effectiveness of detector 1 was demonstrated by numerical simulation in figure \ref{figure4}. Here $n=1024$, $d_1=40,d_2=10$, $\text{SNR}=-10\text{dB}$ and the target and clutter signal were chosen randomly from corresponding subspaces. The average values of 100 monte-carlo simulations of the volume correlation $1/T(\bm{R}^{(m)})$ with respect to different $m$ are plotted, and the values of the detector output with respect to each simulation are showed by a scatter diagram in the small sub-figures. It can be seen from the figures that as $m$ increases, the test quantity $1/T(\bm{R}^{(m)})$ converges to infinity when there is target signal, while there is no target signal, $1/T(\bm{R}^{(m)})$ converges to a finite value. Therefore, as a whole, the simulation result verified the validity of VC subspace detector.

\begin{figure}[htbp]
	\centering
	\includegraphics[width=0.49\textwidth]{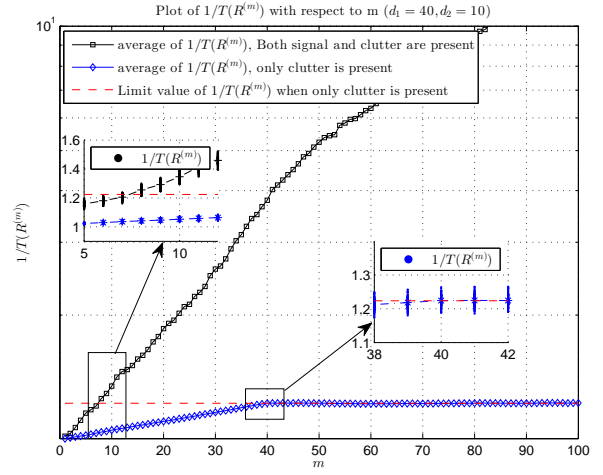}
	\caption{The reciprocal of volume correlation $1/T(\bm R^{(m)})$}
	\label{figure4}
\end{figure}

\section{conclusion}
In this paper, we propose a novel subspace detector that can detect target signal buried in low-rank clutters and random noise without knowledge of clutter subspace. The proposed detector utilize the geometrical relation of clutter and target subspaces. The target signal could be detected without prior learning of clutter structure with our detector. The influence of clutter will be eliminated simultaneously with the detection being performed. Theoretical and numerical analysis has validated the effectiveness of our new detector.

\section*{Acknowledgment}

 This paper is supported by project 61201356 from National Natural Science Foundation of China.



\newpage

\appendices
\small

\section{Proof of Theorem \ref{H1Thm1N}}
\par In order to prove Theorem \ref{H1Thm1N}, two theorems as intermediate results are needed.
\begin{Theorem}\label{H0Thm1}
	Let $\bm H$ be a $n$-dimensional Hilbert space, $\bm{H}_S$ and $\bm{H}_C$ be the subspaces of $H$ corresponding to target and clutter respectively. $\bm{H}_S\cap\bm{H}_C=\{0\}$, $\Dim{\bm{H}_C}=d_1$, $\Dim{\bm{H}_S}=d_2$, Suppose $\bm{y}_i, i=1,\cdots,$ be sampled data either containing both target and clutter,
	\begin{equation}
	\bm{y}_i = \bm{x}^{(S)}_i+\bm{x}^{(C)}_i,\qquad i = 1,2,\cdots
	\end{equation}
	or only containing clutter
	\begin{equation}
	\bm{y}_i = \bm{x}^{(C)}_i,\qquad i = 1,2,\cdots
	\end{equation}	
	where $\bm{x}^{(S)}_i\in\bm{H}_S$ and $\bm{x}^{(C)}_i\in\bm{H}_C$. Let
	\begin{equation}
	\bm{Y}^{(m)}:=[\bm{y}^{(1)},\cdots,\bm{y}^{(m)}],
	\end{equation}
	and $\bm{Q}_{\bm{Y}}^{(m)}$ and $\bm{Q}_S$ be orthogonal matrix with columns being the basis vectors of $\Span(\bm{Y}^{(m)})$ and $\bm{H}_S$, then we have the following monotone property
	\begin{equation}\label{Monotone}
	1/T(\bm{Y}^{(1)})\leq 1/T(\bm{Y}^{(2)})\leq\cdots\leq 1/T(\bm{Y}^{(d_1)})
	\end{equation}
	holds for both scenarios, where 	$T(\bm{Y}^{(m)})=\Vol_{m+d_2}([\bm{Q}_{\bm{Y}}^{(m)},\bm{Q}_S])$.
	
\end{Theorem}

\begin{Theorem}\label{H0Thm2}
	Under the same assumption of theorem \ref{H0Thm1}, the sufficient and necessary condition for existence of target signal in sampled data is that: there exits an integer $k$, such that
	\begin{equation}
	T(\bm{Y}^{(m)})=\Vol_{m+d_2}([\bm{Q}_{\bm{Y}}^{(m)},\bm{Q}_S])=0,\label{corr1}
	\end{equation}
	for all $m > k$, and $k=d_1$. Or equivalently, the sufficient and necessary condition for non-existence of target signal in sampled data is that:
	\begin{equation}
	T(\bm{Y}^{(m)})=\Vol_{m+d_2}([\bm{Q}_{\bm{Y}}^{(m)},\bm{Q}_S])>0,\label{corr1}
	\end{equation}	
	for all $m > 1$.
\end{Theorem}
\subsection{Proof of Theorem \ref{H0Thm1}}

We will prove the monotone property (\ref{Monotone}). An useful Lemma is proved at first. It is just a simple property with very clear geometric intuition for volume of subspaces.

\begin{Lemma}\label{Lemma1}
	Let $\bm{Y}^{(k)}=[\bm{y}_1,\cdots,\bm{y}_k]\in\mathbb{R}^{n\times{k}}$ and $\bm{X}\in\mathbb{R}^{n\times{l}}$ be two matrices, then
	\begin{equation}
	\Vol_{k+l}([\bm{X},\bm{Y}^{(k)}])=\Vol_{k-1+l}([\bm{X},\bm{Y}^{(k-1)}])\|\bm{P}_{[\bm{X},\bm{Y}^{(k-1)}]}^{\perp}\bm{y}_k\|
	\end{equation}
	where $\bm{P}_{\bm{A}}^{\perp}$ is the orthogonal complement of projection matrix on column space of matrix $\bm{A}$.
\end{Lemma}
\begin{proof}
	According to the definition of volume for subspace (\ref{VolumeDef2}), we have
	\begin{align}
	\Vol_{k+l}^2([\bm{X},\bm{Y}^{(k)}])\hspace{6.5cm}\nonumber\\ =\det([\bm{X},\bm{Y}^{(k)}]^{T}[\bm{X},\bm{Y}^{(k)}])\hspace{4.55cm}\nonumber\\
	=\det\left(
	\begin{array}{cc}
	\bm{X}^{T}\bm{X}&\bm{X}^{T}\bm{Y}^{(k)}\\
	(\bm{Y}^{(k)})^T\bm{X}&(\bm{Y}^{(k)})^T\bm{Y}^{(k)}
	\end{array}
	\right)\hspace{3.2cm}\nonumber\\
	=\det\left(
	\begin{array}{cc:c}
	\bm{X}^{T}\bm{X}&\bm{X}^{T}\bm{Y}^{(k-1)}&\bm{X}^{T}\bm{y}^{(k)}\\
	(\bm{Y}^{(k-1)})^T\bm{X}&(\bm{Y}^{(k-1)})^T\bm{Y}^{(k-1)}&(\bm{Y}^{(k-1)})^T\bm{y}_k\\
	\hdashline
	\bm{y}_k^T\bm{X}&\bm{y}_k^T\bm{Y}^{(k-1)}&\bm{y}_k^T\bm{y}_k\\
	\end{array}
	\right)
	\end{align}
	Using Schur complement formula,
	\begin{equation}
	\det\left(
	\begin{array}{cc}
	\bm{A}&\bm{B}\\
	\bm{C}&\bm{D}
	\end{array}
	\right)
	=\det(\bm{A})\det(\bm{D}-\bm{B}\bm{A}^{-1}\bm{C}),
	\end{equation}
	we obtain (in the next page)
	\newcounter{mytempeqncnt}
	\begin{figure*}[!t]
		\normalsize
		\begin{small}
			\begin{align}
			\Vol_{k+l}^2([\bm{X},\bm{Y}^{(k)}])
			=\det\left(
			\begin{array}{cc}
			\bm{X}^{T}\bm{X}&\bm{X}^{T}\bm{Y}^{(k-1)}\\
			(\bm{Y}^{(k-1)})^T\bm{X}&(\bm{Y}^{(k-1)})^T\bm{Y}^{(k-1)}
			\end{array}
			\right)\cdot \hspace{8cm}
			\nonumber \\
			\det\left(
			\bm{y}_k^T\bm{y}_k-
			\bm{y}_k^T[\bm{X},\bm{Y}^{(k-1)}]
			\left(
			\begin{array}{cc}
			\bm{X}^{T}\bm{X}&\bm{X}^{T}\bm{Y}^{(k-1)}\\
			(\bm{Y}^{(k-1)})^T\bm{X}&(\bm{Y}^{(k-1)})^T\bm{Y}^{(k-1)}
			\end{array}
			\right)^{-1}
			[\bm{X},\bm{Y}^{(k-1)}]^T\bm{y}_k
			\right) \nonumber \\
			=\Vol_{k-1+l}^2([\bm{X},\bm{Y}^{(k-1)}])\bm{y}_k^{T}(\bm{I}-\bm{P}_{[\bm{X},\bm{Y}^{(k-1)}]})\bm{y}_k, \hspace{7.6cm}
			\end{align}
		\end{small}
		
	\end{figure*}
	
	
	Because $\bm{I}-\bm{P}_{[\bm{X},\bm{Y}^{(k-1)}]}$ is idempotent matrix,
	\begin{equation}
	(\bm{I}-\bm{P}_{[\bm{X},\bm{Y}^{(k-1)}]})^2=\bm{I}-\bm{P}_{[\bm{X},\bm{Y}^{(k-1)}]},
	\end{equation}
	we obtain
	\begin{align}
	\Vol_{k+l}^2([\bm{X},\bm{Y}^{(k)}])
	&=\Vol_{k-1+l}^2([\bm{X},\bm{Y}^{(k-1)}])\bm{y}_k^{T}(\bm{I}-\bm{P}_{[\bm{X},\bm{Y}^{(k-1)}]})^2\bm{y}_k\nonumber\\
	&=\Vol_{k-1+l}^2([\bm{X},\bm{Y}^{(k-1)}])\|\bm{P}_{[\bm{X},\bm{Y}^{(k-1)}]}^{\perp}\bm{y}_k\|^2
	\end{align}
	This is just what we want to prove.
\end{proof}
\par

With the notations in theorem \ref{H0Thm1}, when $m<\Dim(\bm{H}_S\oplus\bm{H}_C)$, we have
\begin{equation}
T(\bm{Y}^{(m)})=\Vol_{d_2+m}([\bm{Q}_{S}, \bm{Q}_{\bm{Y}}^{(m)}]),
\end{equation}
where $\bm{Q}_{\bm{Y}}^{(m)}=[\bm{q}_{\bm{Y}}^{(1)},\bm{q}_{\bm{Y}}^{(2)},\cdots,\bm{q}_{\bm{Y}}^{(m)}]$ is matrix with columns being orthogonal basis vectors for $\bm{Y}^{(m)}$. Moreover,
\begin{equation}
\bm{Q}_{\bm Y^{(m)}}=[\bm{Q}_{\bm Y^{(m-1)}}, \bm{q}_{\bm{Y}}^{(m)}],
\end{equation}
we have
\begin{align}
\bm{q}_{\bm{Y}}^{(m)}
&=\frac{(\bm{I}-\bm{Q}_{\bm{Y}}^{(m-1)}(\bm{Q}_{\bm{Y}}^{(m-1)})^T)\bm{y}_{m}}
{\|(\bm{I}-\bm{Q}_{\bm{Y}}^{(m-1)}(\bm{Q}_{\bm{Y}}^{(m-1)})^T)\bm{y}_{m}\|}\nonumber\\
&=\bm{P}_{\bm{Q}_{\bm{Y}}^{(m-1)}}^{\perp}\bm{y}_{m}\nonumber\\
&=\bm{P}_{\bm{Y}^{(m-1)}}^{\perp}\bm{y}_{m}.
\end{align}
Using Lemma \ref{Lemma1}, Let $\bm{X}=\bm{Q}_S$, $\bm{Y}=\bm{Q}_{\bm{Y}}^{(m)}$, we obtain
\begin{align}
T(\bm{Y}^{(m)})&=\Vol_{d_2+m-1}([\bm{Q}_{S},\bm{Q}_{\bm{Y}}^{(m-1)}])\|\bm{P}_{[\bm{Q}_S,\bm{Q}_{\bm{Y}}^{(m-1)}]}^{\perp}\bm{q}_{\bm{Y}}^{(m)}\|\\
&=T(\bm{Y}^{(m-1)})\|\bm{P}_{[\bm{Q}_S,\bm{Q}_{\bm{Y}}^{(m-1)}]}^{\perp}\bm{q}_{\bm{Y}}^{(m)}\|.
\end{align}
Take into account the property of projection matrices,
\begin{equation}
\|\bm{P}_{[\bm{Q}_S,\bm{Q}_{\bm{Y}}^{(m-1)}]}^{\perp}\bm{q}_{\bm{Y}}^{(m)}\|\leq\|\bm{q}_{\bm{Y}}^{(m)}\|=1,
\end{equation}
we have
\begin{equation}
T(\bm{Y}^{(m)})\leq{T(\bm{Y}^{(m-1)})}.
\end{equation}
Thus (\ref{Monotone}) and theorem \ref{H0Thm1} has been proven.

\subsection{Proof of Theorem \ref{H0Thm2}}

Let $\bm{H}_S$ and $\bm{H}_C$ be target subspace and clutter subspace respectively. $\Dim(\bm{H}_S)=d_2$, $\Dim(\bm{H}_C)=d_1$. Assume there exists target signal in received data $\{\bm{y}_i, i=1,\cdots,m\}$, that is to say,
\begin{equation}
\bm{y}_i=\bm{H}_S\bm\alpha_i+\bm{H}_C\bm\beta_i,
\end{equation}
for some $i\in{S}\subset\{1,2,\cdots,m\}$, and
\begin{equation}
\bm{y}_i=\bm{H}_C\bm\beta^{(i)},
\end{equation}
for other $i\in\{1,2,\cdots,m\}\setminus{S}$. Then under the generic hypothesis, we have
\begin{equation}
\Rank([\bm{y}_1,\cdots,\bm{y}_m])=m,
\end{equation}
for $m\leq{d_1}$. Hence the result of successive orthogonalization could be written as
\begin{equation}
\bm{Q}_{\bm{Y}}^{(m)}=[\bm{q}_{\bm{Y}}^{(1)},\bm{q}_{\bm{Y}}^{(2)},\cdots,\bm{q}_{\bm{Y}}^{(m)}],
\end{equation}
It should be stressed that in the case of $k=d_1+1$, we still have
\begin{equation}
\Rank([\bm{y}_1,\cdots,\bm{y}_{d_1+1}])=d_1+1,
\end{equation}
because of the presence of target signal. In other words,
\begin{equation}\label{Ortho1}
\bm{Q}_{\bm{Y}}^{(d_1+1)}=[\bm{q}_{\bm{Y}}^{(1)},\bm{q}_{\bm{Y}}^{(2)},\cdots,\bm{q}_{\bm{Y}}^{(d_1)},\bm{q}_{\bm{Y}}^{(d_1+1)}].
\end{equation}
For $\bm{Y}^{(d_1+1)}$, all of its $d_1+1$ linearly independent directions includes two parts, one with $d_1$ directions come from clutter subspace $\bm{H}_C$ and the other one direction is contributed by target subspace $\bm{H}_S$.

Let $\bm{Q}_S$ be the matrix with columns being the orthonormal basis vectors of $\bm{H}_S$,  (\ref{Ortho1}) means that
\begin{equation}
\Rank([\bm{Q}_S,\bm{Q}_{\bm{Y}}^{(d_1+1)}])=d_2+d_1,
\end{equation}
and the number of nonzero columns of $[\bm{Q}_S,\bm{Q}_{\bm{Y}}^{(d_1+1)}]$ is $d_1+d_2+1$. Using Lemma \ref{lemma2}, we obtain
\begin{equation}
\Vol_{d_1+d_2+1}([\bm{Q}_S,\bm{Q}_{\bm{Y}}^{(d_1+1)}])=0,
\end{equation}
Take $k=d_1+1$, the necessary part of theorem has been proved.

On the contrary, under the generic hypothesis, if there exists $k$ such that
\begin{equation}\label{SigEx}
\Vol_{k+d_2}([\bm{Q}_S,\bm{Q}_{\bm{Y}}^{(k)}])=0,
\end{equation}
then there must be target signal in sample data ${\bm{y}_1,\cdots,\bm{y}_k}$.

Assume this was not the case, then each sample $\bm{y}_i$ contained no target signal. we have
\begin{equation}
\bm{Y}^{(k)}=[\bm{y}_1,\cdots,\bm{y}_k]\subset\bm{H}_C,
\end{equation}
therefore the orthonormal basis matrix $\bm{Q}_{\bm{Y}}^{(m)}$ of $\bm{Y}^{{m}}$ satisfied
\begin{equation}
\bm{Q}_{\bm{Y}}^{(m)}=[\bm{q}_{\bm{Y}}^{(1)},\bm{q}_{\bm{Y}}^{(2)},\cdots,\bm{q}_{\bm{Y}}^{(m)}]
\end{equation}
for $m\leq{k_1}$ and
\begin{equation}
\bm{Q}_{\bm{Y}}^{(m)}=[\bm{q}_{\bm{Y}}^{(1)},\bm{q}_{\bm{Y}}^{(2)},\cdots,\bm{q}_{\bm{Y}}^{(k_1)}]
\end{equation}
for $m>k_1$. According to (\ref{CorrVol}) and (\ref{VolAng}), we obtain
\begin{align}
\Vol_{m+k_2}([\bm{Q}_S,\bm{Q}_{\bm{Y}}^{(m)}])
&=\frac{\Vol_{m+k_2}([\bm{Q}_S,\bm{Q}_{\bm{Y}}^{(m)}])}{\Vol_{m}(\bm{Q}_S)\Vol_{k_2}(\bm{Q}_{\bm{Y}}^{(m)})}\nonumber\\
&=\Corr_{\textrm{vol}}(\bm{Q}_S,\bm{Q}_{\bm{Y}}^{(m)})\nonumber\\
&=\prod_{j=1}^{\min(d_1,d_2)}\sin\theta_j(\bm{X}_1, \bm{X}_2) >0.
\end{align}
Considering the monotone relation (\ref{Monotone}), we have
\begin{equation}
\Vol_{m+k_2}([\bm{Q}_S,\bm{Q}_{\bm{Y}}^{(m)}])>0,\quad{\forall m\in\mathbb{N}},
\end{equation}
Contradiction! We have verified the sufficient part and the whole theorem has been proved.

\subsection{Proof of Theorem \ref{H1Thm1N}}

To complete proof of Theorem \ref{H1Thm1N}, several lemmas are required as necessary tools. These lemmas concerned with asymptotical distribution of eigenvectors of sample covariance matrix, concentration bounds and matrix perturbation.

\begin{Lemma}\cite{stoica1989music}\label{stoica}
	Consider the matrix $\bm{\hat{Q}}^{(r)}\in\mathbb{R}^{n\times{r}}$ with columns being the $r$ eigenvectors of sample covariance matrix $\hat{\mathrm{R}}^{(m)}\in\mathbb{R}^{n\times{n}}$ corresponding to the largest $r$ eigenvalues,
	$$
	\bm{\hat{Q}}^{(r)}=[\bm{\hat{q}}_1, \bm{\hat{q}}_2,\cdots,\bm{\hat{q}}_r],
	$$
	its asymptotic distribution (for large $m$) is jointly Gaussian with mean
	$$
	\bm{Q}=[\bm{q}_1,\bm{q}_2,\cdots,\bm{q}_r],
	$$
	and covariance $\bm \Sigma_1^{(m)},\cdots,\bm \Sigma_r^{(m)}$, where
	\begin{equation}
	\bm \Sigma_i^{(m)}:= \frac{\lambda_i}{m}\Big[\sum_{\stackrel{j=1}{j \neq i}}^r\frac{\lambda_j}{(\lambda_i-\lambda_j)^2}\bm{q}_i\bm{q}_i^T + \sum_{j=r+1}^{P}\frac{\sigma^2}{(\sigma^2-\lambda_i)^2}\bm{q}_i\bm{q}_i^T\Big],\nonumber
	\end{equation}
	$i=1,\cdots,r$, and
	\begin{equation}\label{perturbcov}
	\mathbb{E}(\bm{\hat{q}}_i^{(m)}-\bm{q}_i)(\bm{\hat{q}}_k^{(m)}-\bm{q}_k)^T=\bm\Sigma_i^{(m)}\cdot\delta_{i,k},\quad i,k=1,\cdots,r
	\end{equation}
	where $\lambda_1\geq\lambda_2\geq\cdots\lambda_r\geq\lambda_{r+1}=\cdots=\lambda_n=\sigma^2$ are eigenvalues of the covariance matrix $\mathrm{R}_r$ in (\ref{CorrMatrix}), with $\bm{q}_1,\cdots,\bm{q}_n$ the corresponding eigenvectors.
\end{Lemma}

\begin{Lemma}\label{LemmaRMF}
	For the random matrix
	\begin{equation}
	\bm{E} = [\bm{e}_1,\cdots,\bm{e}_r]\in\mathbb{R}^{n\times{r}}
	\end{equation}
	where $\bm{e}_i\thicksim\mathcal{N}(0,\bm\Sigma_i), 1\leq{i}\leq{d}$, and $\mathbb{E}(\bm{e}_i\bm{e}_k^T)=\bm\Sigma_i\cdot\delta_{i,k}$, then for any
	$0<\varepsilon<1$, there exists a constant $C>0$ that depends on $\bm\Sigma_i$, such that
	\begin{equation}\label{concentration}
	\|\bm{E}\|_F^2\leq(1+\varepsilon)\sum_{i=1}^r\Trace(\bm\Sigma_i),
	\end{equation}
	holds with probability
	\begin{equation}\label{prob}
	\mathbb{P}\geq1-\exp\{-\frac{r\cdot{n}\cdot\varepsilon^2}{C}\}.
	\end{equation}
\end{Lemma}

\begin{proof}
	From the definition of Frobenius norm, we know that
	\begin{equation}\label{FNorm}
	\|\bm{E}\|_F^2 = \sum_{i=1}^r\|\bm{e}_i\|_2^2.
	\end{equation}
	For any $1\leq{i}\leq{r}$, $\bm{e}_i\thicksim\mathcal{N}(0,\bm\Sigma_i)$, the eigenvalue decomposition of $\bm\Sigma_i\in\mathbb{R}^{n\times{n}}$ could be written as
	\begin{equation}
	\bm\Sigma_i=\bm{V}_i\bm\Lambda_i\bm{V}_i^T,
	\end{equation}
	where the diagonal matrix $\bm\Lambda_i:=\Diag(\lambda_{i,1}^2,\cdots,\lambda_{i,n}^2)$ and $\lambda_{i,1}^2\geq\lambda_{i,2}^2\geq\cdots\geq \lambda_{i,n}^2\geq0$ are eigenvalues of $\bm\lambda_i$.
	
	Let
	\begin{equation}
	\bm{\tilde e}_i=\bm{V}_i^T\bm{e}_i,
	\end{equation}
	then
	\begin{equation}
	\bm{\tilde e}_i\thicksim\mathcal{N}(0,\bm\Lambda_i),\quad \|\bm{\tilde e}_i \|_2^2=\|\bm{e}_i\|_2^2.
	\end{equation}
	Denote the elements of vector $\bm{\tilde e}_i$ by
	\begin{equation}
	\bm{\tilde e}_i = [\tilde e_{i,1},\cdots,\cdots \tilde{e}_{i,n}]^T,
	\end{equation}
	then different $\tilde e_{i,j}$ are independent and satisfy
	\begin{equation}\label{elemdist}
	\tilde e_{i,j} \thicksim \mathcal{N}(0,\sigma_{i,j}^2),\quad 1\leq j \leq n.
	\end{equation}
	Now we stack all these vectors $\bm{\tilde e}_i, 1 \leq i \leq d$ into a single vector, i.e., we let
	\begin{equation}
	\bm{\tilde e} := [\bm{\tilde e}_1^T, \bm{\tilde e}_2^T, \cdots, \bm{\tilde e}_n^T]^T \in \mathbb{R}^{r\times{n}},
	\end{equation}
	then we have $\bm{\tilde e} \thicksim \mathcal{N}(0, \bm \Lambda), $
	\begin{equation}
	\bm \Lambda = \Diag(\bm \Lambda_1, \cdots, \bm \Lambda_r) = \Diag(\sigma_{1,1}^2,\cdots,\sigma_{1,n}^2,\cdots,\sigma_{r,1}^2,\cdots,\sigma_{r,n}^2). \nonumber
	\end{equation}
	Therefore, (\ref{FNorm}) is equivalent to
	\begin{equation}
	\|\bm E\|_F^2 = \sum_{i=1}^r  \|\bm{\tilde e}_i\|_2^2 = \|\bm{\tilde e}\|_2^2.
	\end{equation}
	
	It is well known that the norm of a Gaussian random vector will concentrate around its expectation. It has been proved that the norm of an i.i.d. Gaussian random vector will concentrate around its expectation\footnote{Chapter 4, \cite{PHDThesis}}. The problem of the concentration of $\|\bm{\tilde e}\|_2^2$ here is only slightly different with the one in \cite{PHDThesis}. In particular, the elements of
	$\bm{\tilde e}$ have different variances in our case and the i.i.d hypothesis is violated to a small extent. Therefore, the proof will be adapted from the proof of Theorem 4.2 in \cite{PHDThesis}. So only the different part will be given in the following proof.
	
	Firstly, we have
	\begin{equation}
	\mathbb{E}\{\|\bm{\tilde e}\|_2^2\} = \sum_{i=1}^r\sum_{j=1}^n \sigma_{i,j}^2 = \sum_{i=1}^r\Trace(\bm \Sigma_i),
	\end{equation}
	
	Then follow the same approach as \cite{PHDThesis} and utilize Markov's Inequality. For any parameter $\beta>0$ and $\lambda>0$, we have
	\begin{eqnarray}
	\mathbb{P}\{\|\bm{\tilde e}\|_2^2 \geq \beta \sum_{i=1}^r \Trace(\bm \Sigma_i)\} = \mathbb{P}\{\exp(\lambda \|\bm{\tilde e}\|_2^2) \geq \exp(\lambda\beta \sum_{i=1}^r \Trace(\bm \Sigma_i)) \} \nonumber \\
	= \prod_{i=1}^r \mathbb{P}\{\exp(\lambda \|\bm{\tilde e}_i\|_2^2) \geq \exp(\lambda\beta  \Trace(\bm \Sigma_i)) \} \hspace{0cm}\nonumber \\
	\leq \prod_{i=1}^r  \frac{\mathbb{E}\{\exp(\lambda \|\bm{\tilde e}_i\|_2^2)\}}{\exp(\lambda\beta \Trace(\bm \Sigma_i))}\hspace{2.5cm} \nonumber \\
	= \prod_{i=1}^r \prod_{j=1}^P [\frac{\mathbb{E}\{\exp(\lambda \tilde e_{i,j}^2)\}}{\exp(\lambda\beta \sigma_{i,j}^2)}].\hspace{2.15cm} \nonumber
	\end{eqnarray}
	The moment generating function of the Gaussian random variable $\tilde e_{i,j}$ is:
	\begin{equation}
	\mathbb{E}\{\exp(\lambda \tilde e_{i,j}^2)\}=\frac{1}{\sqrt{1-2\lambda \sigma_{i,j}^2}},
	\end{equation}
	let
	\begin{equation}
	\sigma_{\max}:=\max_{i,j}\sigma_{i,j},\quad \sigma_{\min} := \min_{i,j}\sigma_{i,j},
	\end{equation}
	we have
	\begin{equation}\label{Obj}
	\mathbb{P}\{\|\bm{\tilde e}\|_2^2 \geq \beta \sum_{i=1}^r\Trace(\bm \Sigma_i)\}\leq\left(\frac{\exp(-2\lambda\beta \sigma_{\min}^2)}{1-2\lambda\sigma_{\max}^2}\right)^{r\cdot{n}/2},\quad\lambda>0,
	\end{equation}
	
	The rest of proof is the same as Theorem 4.2 in \cite{PHDThesis} and will be described briefly. Replacing $\lambda$ with its optimal value such that the right side of (\ref{Obj}) is minimized, and regarding some  formulas involving $\sigma_{\max}$ and $\sigma_{\min}$ for a constant $C$, we can derive the result of this lemma (which is also the result of Corollary 4.1 in \cite{PHDThesis} under i.i.d hypothesis):
	\begin{eqnarray}
	\mathbb{P}\{\|\bm{\tilde e}\|_2^2 \geq (1+\varepsilon)\sum_{i=1}^r\Trace(\bm \Sigma_i)\}\leq\exp (-\frac{r\cdot{n}\cdot\varepsilon^2}{C}), 
	\end{eqnarray}
	holds for any $0<\varepsilon<1$, where $C>0$ is a constant depending on $\sigma_{\max}$ and $\sigma_{\min}$.
\end{proof}

Next, the lemma will be presented to estimate the influence of the error between sample eigenvectors $\bm{\hat{Q}}^{(m)}$ and its true value on the volume correlation computation. Motivated by the relation between volume and determinant, the matrix perturbation theory was utilized to derive the result needed.

\begin{Lemma}(Corollary 2.7 in \cite{ipsen2008perturbation})\label{perturb}
	For the matrix $\bm A \in \mathbb{R}^{n \times n}$, and the perturbation matrix $\bm E \in \mathbb{R}^{n \times n}$, we have
	\begin{itemize}
		\item If $\bm A $ is full-rank, then
		\begin{equation}\label{H0perturb}
		|\det(\bm A+\bm E)-\det(\bm A)| \leq \sum_{i=1}^n s_{n-i}(\bm A)\|\bm E\|_2^i,
		\end{equation}
		\item If $\Rank(\bm A) = k$ for some $1\leq k \leq n-1$, then
		\begin{equation}\label{H1perturb}
		|\det(\bm A+\bm E)| \leq \|\bm E\|_2^{n-k} \sum_{i=0}^k s_{k-i}(\bm A)\|\bm E\|_2^i.
		\end{equation}
		here $s_{k}(\bm A)$ is defined as the $k$th elementary symmetric function of singular values of matrix $\bm{A}\in\mathbb{R}^{n\times{n}}$:
		\begin{equation}\label{ESF}
		s_{k}(\bm A) := \sum_{1 \leq i_1 \leq \cdots \leq i_k \leq n} \sigma_{i_1}\cdots \sigma_{i_k}, 1\leq k \leq n.
		\end{equation}
	\end{itemize}
\end{Lemma}

Now we will prove theorem \ref{H1Thm1N}. Assume the sample data be
\begin{equation}
\bm{R}^{(m)}=[\bm{y}_1,\cdots,\bm{y}_m],
\end{equation}
and its sample correlation matrix be $\mathrm{R}_{\bm y}^{(m)}=\frac{1}{m}\bm{R}^{(m)}(\bm{R}^{(m)})^T$, then the volume correlation is of the form of:
\begin{equation}
T(\bm{R}^{(m)})=\Vol_{d_2+k_m}([\bm{Q}_S,\bm{\hat{Q}}^{(m)}]),
\end{equation}
where $\bm{Q}_S$ is the matrix with columns being the orthonormal basis vectors of target subspace $\bm{H}_S$ and $\bm{\hat{Q}}^{(m)}$ is the matrix with columns being the eigenvectors of $\mathrm{R}_{\bm y}^{(m)}$ corresponding to relatively large eigenvalues. Without lossing generality, assume
\begin{equation}
\bm{\hat{Q}}^{(m)}=[\bm{\hat{q}}_1,\cdots,\bm{\hat{q}}_{k_m}],
\end{equation}
We have
\begin{align}\label{detperturb}
T(\bm R^{(m)})= \Vol_{d_2+k_m}([\bm{Q}_{S}, \bm{\hat{Q}}^{(m)}])\hspace{4.5cm}\nonumber\\
= {\det}^{1/2}\left(\bm Q_{S}^T\bm Q_{S}\right)
{\det}^{1/2}(\bm I_n - (\bm{\hat{Q}}^{(m)})^T\bm Q_{S}\left(\bm{Q}_{S}^T\bm{Q}_{S}\right)^{-1}\bm{Q}_{S}^T\bm{\hat{Q}}^{(m)})\nonumber\\
= {\det}^{1/2}\left((\bm{\hat{Q}}^{(m)})^T\bm P_{S}^{\perp}\bm{\hat{Q}}^{(m)}\right).\hspace{5.cm}
\end{align}
where $\bm{P}_{S}^{\perp}$ is the orthogonal complement of the projection matrix onto the target subspace $\bm{H}_S$.

We noticed that approximate eigenvectors $\bm{\hat{Q}}^{(m)}$ in (\ref{detperturb}) should be replaced by its true value to obtain the conclusion of theorem \ref{H1Thm1N}. It is natural for using lemma \ref{stoica} to estimate the error of approximation because the main object of calculation involved in (\ref{detperturb}) is determinant. Firstly assume $\bm{\hat{Q}}^{(m)}$ could be expressed as following linear random perturbation model,
\begin{equation}
\bm{\hat{Q}}^{(m)} = \bm{Q}_Y^{(m)}+\bm{E}^{(m)},
\end{equation}
where
\begin{equation}
\bm E^{(m)} = [\bm e_1^{(m)},\cdots,\bm e_{k_m}^{(m)}],
\end{equation}
with $\bm e_i^{(m)}\thicksim \mathcal{N}(0,\bm \Sigma_i^{(m)})$, and ${\bm e_i^{(m)}}$ are mutual independent. Then according to (\ref{perturbcov}), we have
\begin{equation}
\bm \Sigma_i^{(m)} = \frac{\lambda_i}{m} \Big[ \sum_{\stackrel{j=1}{j \neq i}}^{k_m} \frac{\lambda_j}{(\lambda_i-\lambda_j)^2} \bm q_i \bm q_i^T + \sum_{j=k_m+1}^{n}\frac{\sigma^2}{(\sigma^2-\lambda_i)^2}\bm q_i \bm q_i^T \Big].
\end{equation}
Therefore (\ref{detperturb}) becomes
\begin{align}\label{detperturb2}
T(\bm R^{(m)})
={\det}^{1/2}\left((\bm{\hat{Q}}^{(m)})^T\bm P_{S}^{\perp}\bm{\hat{Q}}^{(m)}\right) \hspace{3cm}\nonumber \\
={\det}^{1/2}\left((\bm P_{S}^{\perp}\bm{Q}_Y^{(m)}+\bm P_{S}^{\perp}\bm E^{(m)})^T(\bm P_{S}^{\perp}\bm{Q}_Y^{(m)}+\bm P_{S}^{\perp}\bm E^{(m)})\right),
\end{align}
for simplicity, let
$$
\bm V = \bm P_{S}^{\perp}\bm{Q}_Y^{(m)}, \qquad \bm W = \bm P_{S}^{\perp}\bm E^{(m)},
$$
then (\ref{detperturb2}) becomes
\begin{equation}
T(\bm R^{(m)})^2 = \det\left((\bm V+ \bm W)^T(\bm V+ \bm W)\right),
\end{equation}
Let
$$
\bm A = \bm V^T \bm V, \qquad \bm E = \bm V^T \bm W + \bm W^T \bm V + \bm W^T \bm W,
$$
then we have
$$
T(\bm R^{(m)})^2 = \det(\bm A + \bm E),
$$
where
\begin{align}\label{RankA}
\bm{A}&=(\bm{Q}_Y^{(m)})^T\bm{P}_S^{\perp}\bm{Q}_Y^{(m)},\\
\bm{E}&=(\bm{Q}_Y^{(m)})^T\bm{P}_S^{\perp}\bm{E}^{(m)}+(\bm{E}^{(m)})^T\bm{P}_S^{\perp}\bm{Q}_Y^{(m)}+(\bm E^{(m)})^T\bm{P}_S^{\perp}\bm E^{(m)},
\end{align}

From lemma \ref{LemmaRMF}, we noted that its two conclusions were distinguished by the rank of matrix $\bm{A}$. It indicated that the rank of $\bm{A}$ was a critical factor for accuracy of approximation. In fact, it determined the infinitesimal order for error of approximation. So the rank of $\bm{A}$ should be analyzed.

According to (\ref{RankA}), we have
\begin{equation}\label{RankB}
(\bm{Q}_Y^{(m)})^T\bm{P}_S^{\perp}\bm{Q}_Y^{(m)}=\bm{I}_{k_m}-(\bm{Q}_Y^{(m)})^T\bm{P}_S\bm{Q}_Y^{(m)}.
\end{equation}
The rank of $\bm{A}$ is closely related to the rank of $\bm{Q}_Y^{(m)}$, in other word, the structure of subspace spanned by received data $\bm{Y}^{(m)}$. There are two possibilities for the structure of $\bm{Q}_Y^{(m)}$,
\begin{itemize}
	\item If the target signal presents, then $\Span(\bm{Q}_Y^{(m)})\cap\bm{H}_S\neq0$ and $\Span(\bm{Q}_Y^{(m)})\cap\bm{H}_C\neq0$;
	\item If the target signal doesn't present, then $\Span(\bm{Q}_Y^{(m)})\cap\bm{H}_S=0$ and $\Span(\bm{Q}_Y^{(m)})\cap\bm{H}_C\neq0$;
\end{itemize}

In the first case, because $\Span(\bm{Q}_Y^{(m)})\cap\bm{H}_S\neq\{0\}$, it is assumed that $k_m^S$ of $k_m$ columns of $\bm{Q}_Y^{(m)}$ were contributed by target subspace and the others came from clutter subspace. The corresponding result on the rank of $(\bm{Q}_Y^{(m)})^T\bm{P}_S^{\perp}\bm{Q}_Y^{(m)}$ can be summarized in the following lemmas
\begin{Lemma}
	Under the hypothesis of existence of target signal in sample data, we have
	\begin{equation}\label{Rank1}
	\Rank((\bm{Q}_Y^{(m)})^T\bm{P}_S^{\perp}\bm{Q}_Y^{(m)})=k_m-k_m^S,
	\end{equation}
	where ${1\leq{k_m}\leq\min(m,d_1+d_2)},{1\leq{k_m^S}\leq\min(k_m,d_2)}$,and all of its nonzero singular values (eigenvalues) are $1$.
\end{Lemma}
\begin{proof}
	It should be noted firstly that $k_m$ can not excess $d_1+d_2$ which is the intrinsic dimension of $\bm{H}_S\oplus\bm{H}_C$, and $k_m^S$ should not be larger than $d_1$ no matter how large the number $m$ of sample data is. Generically, we have $k_m\leq k_{m+1}$. Once $k_{m'}=d_1+d_2$, we have $k_{n}=d_1+d_2$ when $n>k_{m'}$.
	
	It is natural to assume $\bm{Q}_Y^{(m)}=[\bm{\bar{Q}}_S,\bm{\bar{Q}}_S^{\perp}]\bm{B}^{(m)}\in\mathbb{R}^{n\times{k_m}}$, where $\bm{\bar{Q}}_S\in\mathbb{R}^{n\times{k_m^{S}}}$ is matrix with columns being parts of orthonormal basis vectors for $\bm{H}_S$, and $\bm{\bar{Q}}_S^{\perp}\in\mathbb{R}^{n\times(k_m-k_m^{S})}$ is matrix with columns composed of vectors in $\bm{P}_S^{\perp}(\bm{H}_S\oplus\bm{H}_C)$. and $\bm{B}^{(m)}\in\mathbb{R}^{k_m\times{k_m}}$ is a orthogonal matrix. Hence from (\ref{RankB}) we have
	\begin{align}
	(\bm{Q}_Y^{(m)})^T\bm{P}_S^{\perp}\bm{Q}_Y^{(m)}=\bm{I}_{k_m}-(\bm{Q}_Y^{(m)})^T\bm{P}_S\bm{Q}_Y^{(m)}\hspace{3.1cm}\nonumber\\
	=\bm{I}_{k_m}-(\bm{B}^{(m)})^T\left[\begin{array}{c}\bm{\bar{Q}}_S^T\\(\bm{\bar{Q}}_S^{\perp})^T\end{array}\right]\bm{P}_S[\bm{\bar{Q}}_S,\bm{\bar{Q}}_S^{\perp}]\bm{B}^{(m)}\nonumber\\
	=\bm{I}_{k_m}-(\bm{B}^{(m)})^T\left[
	\begin{array}{cc}
	\bm{\bar{Q}}_S^T\bm{P}_S\bm{\bar{Q}}_S&\bm{\bar{Q}}_S^T\bm{P}_S\bm{\bar{Q}}_S^{\perp}\nonumber\\
	(\bm{\bar{Q}}_S^{\perp})^T\bm{P}_S\bm{\bar{Q}}_S&(\bm{\bar{Q}}_S^{\perp})^T\bm{P}_S\bm{\bar{Q}}_S^{\perp}
	\end{array}
	\right]\bm{B}^{(m)}
	\end{align}
	because
	\begin{equation}
	\bm{P}_S\bm{\bar{Q}}_S=\bm{\bar{Q}}_S,\qquad\bm{P}_S\bm{\bar{Q}}_S^{\perp}=0,
	\end{equation}
	we have
	\begin{align}
	\bm{\bar{Q}}_S^T\bm{P}_S\bm{\bar{Q}}_S&=\bm{I}_{k_m^{S}},\nonumber\\
	\bm{\bar{Q}}_S^T\bm{P}_S\bm{\bar{Q}}_S^{\perp}
	&=(\bm{\bar{Q}}_S^{\perp})^T\bm{P}_S\bm{\bar{Q}}_S
	=(\bm{\bar{Q}}_S^{\perp})^T\bm{P}_S\bm{\bar{Q}}_S^{\perp}=0,
	\end{align}
	therefore
	\begin{align}
	(\bm{Q}_Y^{(m)})^T\bm{P}_S^{\perp}\bm{Q}_Y^{(m)}
	&=\bm{I}_{k_m}-(\bm{B}^{(m)})^T\left[
	\begin{array}{cc}\bm{I}_{k_m^{S}}&0\\0&0\end{array}
	\right]\bm{B}^{(m)}\\
	&=(\bm{B}^{(m)})^T\left[
	\begin{array}{cc}0&0\\0&\bm{I}_{k_m-k_m^{S}}\end{array}
	\right]\bm{B}^{(m)}
	\end{align}
	Let
	\begin{equation}
	\bm{B}^{(m)}=\left[\begin{array}{c}\bm{B}_1^{(m)}\\\bm{B}_2^{(m)}\end{array}\right],
	\quad{\bm{B}_1^{(m)}\in\mathbb{R}^{{k_m^S}\times k_m},\ \bm{B}_2^{(m)}\in\mathbb{R}^{(k_m-k_m^S)\times k_m}},
	\end{equation}
	then we have
	\begin{equation}
	(\bm{Q}_Y^{(m)})^T\bm{P}_S^{\perp}\bm{Q}_Y^{(m)}=(\bm{B}_2^{(m)})^T\bm{B}_2^{(m)},
	\end{equation}
	assume the singular value decomposition of $\bm{B}_2^{(m)}$ be
	\begin{equation}
	\bm{B}_2^{(m)}=\bar{\bm{U}}[\bm{\Lambda},\bm{0}]\bar{\bm{V}},
	\end{equation}
	where $\bar{\bm{U}}\in\mathbb{R}^{(k_m-k_m^S)\times(k_m-k_m^S)}$ and $\bar{\bm{V}}\in\mathbb{R}^{k_m\times{k_m}}$ are orthogonal matrices and $\bm{\Lambda}$ is diagonal matrix. because of the orthogonality of $\bm{B}^{(m)}$,
	\begin{equation}
	\bm{B}_2^{(m)}(\bm{B}_2^{(m)})^T=\bm{\bar{U}}\bm{\Lambda}^2\bm{\bar{U}}^T=\bm{I}_{k_m-k_m^S},
	\end{equation}
	hence we have $\bm{\Lambda}=\bm{I}_{k_m-k_m^S}$ and
	\begin{equation}
	(\bm{B}_2^{(m)})^T\bm{B}_2^{(m)}=\bm{\bar{V}}\left[\begin{array}{cc}\bm{I}_{k_m-k_m^S}&\bm{0}\\\bm{0}&\bm{0}\end{array}\right]\bm{\bar{V}}^T,
	\end{equation}
	Then the following conclusion could be drawn: If target signal presents in sample data, then we have
	\begin{equation}
	\Rank((\bm{Q}_Y^{(m)})^T\bm{P}_S^{\perp}\bm{Q}_Y^{(m)})=k_m-k_m^S,
	\end{equation}
	and all of its non-zero singular values (eigenvalues) are $1$.
\end{proof}

\textbf{Proof of Theorem \ref{H1Thm1N}:}\par
\begin{proof}
	Now we are ready to prove Theorem \ref{H1Thm1N}.
	According to (\ref{H1perturb}) in Lemma \ref{perturb} and (\ref{RankA}), we have
	\begin{align}\label{OH1}
	T(\bm{R}^{(m)})^2=&\det(\bm A + \bm E)\nonumber \\
	\leq&\|\bm E\|_2^{k_m-k_m+k_m^S}\sum_{i=0}^{k_m-k_m^S}s_{k_m-k_m^S-i}(\bm{A})\|\bm E\|_2^i, \nonumber \\
	=&s_{k_m-k_m^S}(\bm A)\|\bm E\|_2^{k_m^S}+O(\|\bm E\|_2^{k_m^S+1}),
	\end{align}
	where
	\begin{eqnarray}\label{perturbE}
	\|\bm E\|_2 = \hspace{7cm}\nonumber \\
	\| (\bm{Q}_Y^{(m)})^T \bm{P}_S^{\perp}\bm E^{(m)} + (\bm E^{(m)})^T\bm{P}_S^{\perp} \bm{Q}_Y^{(m)} +  (\bm E^{(m)})^T \bm{P}_S^{\perp}\bm E^{(m)}\|_2 \nonumber \\
	\leq  2 \|(\bm{Q}_Y^{(m)})^T \bm{P}_S^{\perp}\bm E^{(m)}\|_2 + \|\bm{P}_S^{\perp}\bm E^{(m)}\|_2^2 \nonumber \\
	\leq  2 \|\bm{Q}_Y^{(m)}\|_2 \|\bm{P}_S^{\perp}\bm E^{(m)}\|_2 + \|\bm{P}_S^{\perp}\bm E^{(m)}\|_2^2 \nonumber \\
	= 2  \|\bm{P}_S^{\perp}\bm E^{(m)}\|_2 + \|\bm{P}_S^{\perp}\bm E^{(m)}\|_2^2 \hspace{1.2cm}\nonumber \\
	\leq 2  \|\bm{P}_S^{\perp}\bm E^{(m)}\|_F + \|\bm{P}_S^{\perp}\bm E^{(m)}\|_F^2.\hspace{1cm}
	\end{eqnarray}
	\par Then, according to Lemma \ref{stoica} and Lemma \ref{LemmaRMF}, we have
	\begin{equation}
	\bm{P}_S^{\perp}\bm E^{(m)} \thicksim \mathcal{N}(0, \bm{P}_S^{\perp} \bm \Sigma_i^{(m)}(\bm{P}_S^{\perp})^T),
	\end{equation}
	and for any $\varepsilon >0$
	\begin{equation}\label{ProbIneq}
	\|\bm{P}_S^{\perp}\bm E^{(m)}\|_F^2 \leq (1+\varepsilon)\sum_{i=1}^{k_m} \Trace(\bm{P}_S^{\perp}\bm \Sigma_i(\bm{P}_S^{\perp})^T),
	\end{equation}
	holds with probability
	\begin{equation}
	\mathbb{P} \geq 1- \exp\{-\frac{k_m\cdot{n}\cdot\varepsilon^2}{C}\} . \nonumber
	\end{equation}
	
	Consider the right side of (\ref{ProbIneq}), we have
	\begin{eqnarray}
	\sum_{i=1}^{k_m}  \Trace(\bm{P}_S^{\perp}\bm \Sigma_i\bm{P}_S)  = \sum_{i=1}^{k_m} \bigg( \frac{1}{m} \Big( \sum_{\stackrel{j=1}{j\neq i}}^{k_m} \frac{\lambda_i \lambda_j}{(\lambda_i-\lambda_j)^2} \Trace(\bm{P}_S^{\perp} \bm u_j \bm u_j^T \bm{P}_S^{\perp T})  \nonumber \\
	+ \sum_{j=k_m+1}^{n}\frac{\lambda_i\sigma^2}{(\sigma^2-\lambda_i)^2} \Trace( \bm{P}_S^{\perp} \bm u_j \bm u_j^T \bm{P}_S^{\perp T}) \Big) \bigg),\hspace{1cm}\nonumber
	\end{eqnarray}
	because
	\begin{equation}
	\Trace(\bm{P}_S^{\perp} \bm u_j \bm u_j^T \bm{P}_S^{\perp T}) = \Trace( \bm u_j^T \bm{P}_S^{\perp} \bm u_j) \leq 1,
	\end{equation}
	we have
	\begin{eqnarray}\label{TraceIneq}
	\sum_{i=1}^{k_m}\Trace(\bm{P}_S^{\perp}\bm\Sigma_i\bm{P}_S)\leq \hspace{4.5cm}\nonumber \\\frac{1}{m}\left(\sum_{i=1}^{k_m} \sum_{\stackrel{j=1}{j\neq i}}^{k_m} \frac{\lambda_i \lambda_j}{(\lambda_i-\lambda_j)^2} +\sum_{i=1}^{k_m}(n-k_m)\frac{\lambda_i\sigma^2}{(\sigma^2-\lambda_i)^2}\right).
	\end{eqnarray}
	Combine (\ref{TraceIneq}) and (\ref{ProbIneq}), we have for any $\varepsilon>0$ and $0 \leq \delta < 1$, if
	\begin{eqnarray}
	\frac{1}{m}\left(  \sum_{i =1}^{k_m} \sum_{\stackrel{j=1}{j\neq i}}^{k_m} \frac{\lambda_i \lambda_j}{(\lambda_i-\lambda_j)^2}  + \sum_{i =1}^{k_m} (n-k_m)\frac{\lambda_i\sigma^2}{(\sigma^2-\lambda_i)^2} \right) \hspace{1cm}\nonumber \\
	\leq \frac{(\sqrt{\delta+1}-1)^2}{1+\varepsilon},
	\end{eqnarray}
	or equivalently,
	\begin{eqnarray}
	m \geq\frac{1+\varepsilon}{(\sqrt{\delta+1}-1)^2}\cdot \hspace{5cm}\nonumber \\
	\left( \sum_{i =1}^{k_m} \sum_{\stackrel{j=1}{j\neq i}}^{k_m} \frac{\lambda_i \lambda_j}{(\lambda_i-\lambda_j)^2}  + \sum_{i =1}^{k_m} (n-k_m)\frac{\lambda_i\sigma^2}{(\sigma^2-\lambda_i)^2} \right),
	\end{eqnarray}
	then
	\begin{equation}\label{concent1}
	\|\bm{P}_S^{\perp}\bm E^{(m)}\|_F^2 \leq (\sqrt{\delta+1}-1)^2,
	\end{equation}
	holds with probability
	\begin{equation}
	\mathbb{P} \geq 1- \exp\{-\frac{k_m\cdot{n}\cdot\varepsilon^2}{C}\} . \nonumber
	\end{equation}
	Then combining (\ref{perturbE}) with (\ref{concent1}), we get
	\begin{equation}
	\|\bm E\|_2\leq2(\sqrt{\delta+1}-1)+(\sqrt{\delta+1}-1)^2=\delta-1\leq{\delta},
	\end{equation}
	thus we have
	\begin{equation}
	T(\bm R^{(m)})^2\leq s_{k_m-k_m^S}((\bm{Q}_Y^{(m)})^T\bm{P}_S^{\perp}\bm{Q}_Y^{(m)})\delta^{k_m^S} + O(\delta^{k_m^S+1}),
	\end{equation}
	holds with overwhelming probability.
	
	Furthermore, according to the definition of elementary symmetric function of singular values in (\ref{ESF}) and (\ref{Rank1}), it can be verified easily that
	\begin{equation}
	s_{k_m-k_m^S}((\bm{Q}_Y^{(m)})^T\bm{P}_S^{\perp}\bm{Q}_Y^{(m)})=1,
	\end{equation}
	hence we have
	\begin{equation}
	T(\bm R^{(m)})^2\leq \delta^{k_m^S} + O(\delta^{k_m^S+1}),
	\end{equation}
	
	If the number $m$ of sample data is large sufficiently, then we have
	\begin{equation}
	k_m=d_1+d_2,\qquad{k_m^S=d_2},
	\end{equation}
	therefore
	\begin{equation}
	T(\bm R^{(m)})^2\leq \delta^{d_2} + O(\delta^{d_2+1}),
	\end{equation}
	
	On the other hand, using (\ref{H0perturb}) in lemma \ref{perturb}, we can similarly obtain the corresponding result for non-target scenario.
	\begin{equation}\label{OH0}
	|\det(\bm A+\bm E)-\det(\bm A)| \leq  s_{k_m-1}(\bm A)\|\bm E\|_2 + O(\|\bm E\|_2^2),
	\end{equation}
	When $m$ is large sufficiently, we have $k_m=d_1$ and $\bm{Q}_Y^{(m)}=\bm{Q}_C$. According to (\ref{detperturb}), we have
	\begin{equation}
	\det((\bm{Q}_Y^{(m)})^T\bm{P}_S^{\perp}\bm{Q}_Y^{(m)})=\Vol_{d_1+d_2}^2([\bm{Q}_S,\bm{Q}_C]): = \tau^2(\bm{H}_S,\bm{H}_C),\nonumber
	\end{equation}
	Similar to discussion above, for any $0 \leq \delta <1$ and $\varepsilon >0$, if
	\begin{eqnarray}
	m \geq\frac{1+\varepsilon}{(\sqrt{\delta+1}-1)^2}\cdot \hspace{5cm}\nonumber \\
	\left( \Big( \sum_{i =1}^{d_1} \sum_{\stackrel{j=1}{j\neq i}}^{d_1} \frac{\lambda_i \lambda_j}{(\lambda_i-\lambda_j)^2}+\sum_{i =1}^{d_1}(n-d_1)\frac{\lambda_i\sigma^2}{(\sigma^2-\lambda_i)^2}  \Big)\right),
	\end{eqnarray}
	then
	\begin{equation}
	\|\bm{P}_S^{\perp}\bm E^{(m)}\|_F^2 \leq (\sqrt{\delta+1}-1)^2,
	\end{equation}
	hold with probability
	\begin{equation}
	\mathbb{P} \geq 1-\exp\{-\frac{d_1\cdot{n}\cdot \varepsilon^2}{C}\}
	\end{equation}
	therefore
	\begin{equation}
	|T(\bm R^{(m)})^2-\tau^2(\bm{H}_S,\bm{H}_C)|\leq s_{d_1-1}(\bm{Q}_C^T\bm{P}_S^{\perp}\bm{Q}_C)\delta + O(\delta^{2}),
	\end{equation}
\end{proof}

\bibliography{IEEEabrv,CompressiveSensingMultiuserDetection.bib}
\bibliographystyle {IEEEtran}



%

\end{document}